\documentclass[11pt]{article}
\usepackage{times}
\usepackage{amsmath, amssymb, amsthm}
\usepackage{boxedminipage}
\usepackage{algorithmic}
\usepackage{algorithm}

\usepackage{fullpage}

\newenvironment{proofof}[1]{{\noindent \em Proof of #1.  }}{\hfill\qed}

\newcommand{\bR}{\mathbb{R}}

\def\argmax{\mathop{\rm argmax}}

\newtheorem{theorem}{Theorem}[section]

\newtheorem{lemma}[theorem]{Lemma}

\newcommand\ignore[1]{}

\theoremstyle{remark}
\newtheorem{remark}[theorem]{Remark}

\newcommand{\initOneLiners}{%
    \setlength{\itemsep}{0pt}
    \setlength{\parsep }{0pt}
    \setlength{\topsep }{0pt}
}

\theoremstyle{definition}
\newtheorem{definition}[theorem]{Definition}

\begin{document}

\title{Efficiently Learning from Revealed Preference}
%
%
\author{Morteza Zadimoghaddam\thanks{MIT, CSAIL, \texttt{morteza@csail.mit.edu}} \and Aaron Roth\thanks{University of Pennsylvania, \texttt{aaroth@cis.upenn.edu}.}
}
%
%
%

\maketitle              

%

\begin{abstract}
In this paper, we consider the revealed preferences problem from a learning perspective. Every day, a price vector and a budget is drawn from an unknown distribution, and a rational agent buys his most preferred bundle according to some unknown utility function, subject to the given prices and budget constraint. We wish not only to find a utility function which rationalizes a finite set of observations, but to produce a hypothesis valuation function which accurately predicts the behavior of the agent in the future. We give efficient algorithms with polynomial sample-complexity for agents with linear valuation functions, as well as for agents with linearly separable, concave valuation functions with bounded second derivative.
\end{abstract}

\section{Introduction}
Consider the problem of a market-researcher attempting to divine the preferences of a population of consumers merely by observing their past buying behavior. Suppose, for example, that the researcher may observe a consumer each day: every day, the consumer is faced with the choice to buy some subset of goods, each of which may have a different price. The consumer is facing an optimization problem -- each day he attempts to buy the subset of goods that maximizes his utility function, given his budget constraints. The market-researcher, on the other hand, is facing a learning problem. Based on his observations of the consumer, he would like to learn a model for the agent's utility function that can explain his behavior, and that can be used to predict (and therefore optimally exploit) his future behavior.

This is the ``revealed preferences'' problem, and it has received a great deal of attention in the economics literature (see, e.g., \cite{Var05} for a nice survey). Typically, however, the work on the revealed preferences problem has focused on determining whether a set of observations is \emph{rationalizable} or not -- i.e. whether it is consistent with \emph{any} utility function that is monotone increasing in each good. A classic result in this literature is Afriat's Theorem, which roughly states that any finite set of observations is rationalizable if and only if it is rationalizable by a monotone increasing, piecewise linear, concave utility function.

Note, however, that the problem of \emph{rationalizing} is easier than the problem of \emph{learning}. To rationalize a set of observations, it is sufficient to find a utility function which explains past behavior. Learning, however, requires finding a utility function which not only explains past behavior, but also will be \emph{predictive} of future behavior! In particular, Afriat's theorem can be taken as showing that attempting to learn from the set of all monotone increasing, piecewise linear, concave utility functions is as hard (and as hopeless) as learning from the set of all utility functions. Indeed, Beigman and Vohra \cite{BV06} have shown that this class of functions has infinite fat-shattering dimension, and so without further restricting the set of allowable utility functions, no accurate predictions can in general be made after any finite set of observations, even by inefficient learning algorithms!

In this paper, we initiate the study of \emph{efficiently} (in terms of both computational complexity and sample complexity) learning utility functions which can accurately predict \emph{future} purchases of a utility-maximizing agent, given access to past purchase behavior. We necessarily restrict the class of agent utility functions, and consider both linear utility functions, and linearly separable concave utility functions with bounded 2nd derivative. We give polynomial upper and lower bounds on the sample complexity (i.e. the number of observations) required for learning, as well as efficient algorithms that can learn predictive models from polynomially many observations.

\subsection{Our Results}
We consider a model in which an agent has an unknown utility function over a set of $n$ divisible goods. We get to observe the behavior of the agent, who every day faces a set of prices for each good, together with a budget constraint, which is drawn from a fixed but unknown probability distribution. The agent selects a bundle of goods to buy so as to maximize his utility function subject to his budget constraint, and the goal of a learning algorithm is to impute a model for his utility function that correctly predicts his behavior with high probability on future price/budget instances drawn from the same distribution.

We consider both linear utility functions, and then more generally, linearly separable concave utility functions with bounded derivatives. For both of these cases, we give efficient learning algorithms with polynomially bounded sample complexity. We then consider a relaxed model in which our algorithm receives expanded feedback from the agent during the learning stage, and is permitted to predict bundles that are within a small additive error of the agent's optimal bundle. In this relaxed model, we give a polynomial time learning algorithm with improved sample complexity bounds.
\subsection{Related Work}
Work on the ``revealed preferences problem'' has a long history in economics, beginning with the seminal work of Samuelson \cite{Sam38}. Modern work on revealed preferences, in which explanatory utility functions are constructively generated from finitely many agent price/purchase observations began with Afriat \cite{Afr65, Afr67} who showed (via an algorithmic construction) that any finite sequence of observations is rationalizable if and only if it is rationalizable by a piecewise linear, monotone, concave utility function. We will not attempt to review the extremely large body of work on revealed preferences, and instead refer the reader to an excellent survey of Varian \cite{Var05}.

Algorithms that constructively generate utility functions given a finite set of observations can be viewed as \emph{learning algorithms} for the set of all monotone increasing utility functions. These algorithms typically come with a caveat, however, that the hypothesis utility functions they generate have the same description length as the set of observations that they were generated from, and so tend to overfit the data -- this observation is related to a recent paper of Echenique, Golovin, and Wierman \cite{EGW11}, who gave a thought-provoking result: that any set of rationalizable observations can in fact be rationalized by a utility function which is computationally easy to optimize.   However, such a utility function clearly \emph{cannot} be predictive of the future behavior of an agent who is in fact making his decisions based on an intractable utility function, because the hypothesis produced by the learning algorithm would itself be witness to the existence of a polynomially sized circuit for optimizing the purportedly intractable utility function of the agent.

Most related to our work is the work of Beigman and Vohra \cite{BV06} who first pose the revealed preferences problem in the model of computational learning theory, with a distribution over observations and the explicit goal of producing a predictive hypothesis. They show that the set of all monotone utility functions has infinite fat-shattering dimension, and therefore prove that (without restricting the class of allowable utility functions), there does not exist any algorithm (independent of computational efficiency) which can provide any non-trivial predictive guarantees from any finite number of samples, over every distribution over observations. They also show that if the agent utility functions satisfy a certain bounded-jump condition, then the resulting class in fact has finite fat-shattering dimension, and that predictive learning is therefore possible using a finite number of samples. We continue this line of work by considering specific, simple classes of utility functions, and give efficient learning algorithms together with small polynomial upper and lower bounds on the sample complexity necessary for learning.

A very nice recent line of work by Balcan and Harvey, and Balcan et al. \cite{BH11, BCIW11} considers a related problem of learning valuation functions. This is similar in motivation, but is orthogonal to the revealed preference setting considered here because it uses direct access to the valuation function evaluated on bundles, rather than only the ``revealed'' preference of the user, which is the maximum value bundle selected subject to some cost constraint.
\section{Preliminaries}
We consider the \emph{revealed preferences problem} for an agent who when faced with a set of prices over $n$ goods $[n]$ buys the most valued bundle available to him.  A \emph{bundle} of goods is a vector of quantities $x \in [0,1]^n$, one for each good: $x_i$ represents the fraction of the good $i$ that is in the bundle. The goods are \emph{divisible}: i.e. bundles can be arbitrary real valued vectors $x \in [0,1]^n$.

The agent has a \emph{value function} $v:[0,1]^n\rightarrow \bR$. His value for a bundle $x \in [0,1]^n$ is simply $v(x)$. Goods can also be paired with vectors of non-negative \emph{prices} $p \in \bR_+^n$, where $p_i$ is the price for good $i$. The price of a bundle is linear in the goods in the bundle. The price of a bundle $x$ with respect to prices $p$ is therefore simply $ x\cdot p$. Prices are important, because the agent may be faced with a budget constraint $B$: he can only buy bundles $x$ such that $x\cdot p \leq B$.

The agent is a utility maximizer. When faced with a price vector $p$ and a budget $B$, he will choose to buy the bundle that maximizes his value subject to his budget constraint: That is, he will choose the bundle:
 $$x^*(v,p,B) = \argmax_{x \in [0,1]^n : x \cdot p \leq B}v(x)$$

 We  will consider several types of value functions in this paper. A \emph{linear} value function $v$ is defined by a vector $v \in \bR_+^n$, where $v_i$ is the marginal value of good $i$. In this case, $v(x) = v\cdot x$. More generally, we can consider linearly separable concave utility functions. A value function $v$ is linearly separable and concave if it can be described using concave functions $v_1,\ldots,v_n$ where each $v_i:[0,1]\rightarrow \bR_+$ is a one-dimensional real valued function, and we can evaluate $v(x) = \sum_{i=1}^nv_i(x_i)$.

The \emph{revealed preferences problem} is to recover a value function that can explain a sequence of choices that the agent was observed to make. In this paper, we wish to recover a value function that can not only rationalize observed behavior, but can help predict future behavior. In order for this to be a meaningful task, we must assume that the choices presented to the agent are drawn from some distribution.

\begin{definition}
An \emph{example} is a price vector $p \in \bR_+^n$ paired with a budget $B \in \bR_+$. A distribution over examples $\mathcal{D}$ is simply a distribution over $(p,B) \sim [0,1]^n\times \bR_+$.
\end{definition}

\begin{definition}
An observation of an agent with value function $v$,  $(p,B,x^*(p,B,v)) \in \bR_+^n\times\bR_+\times\bR_+^n$ is simply a triple consisting of a price vector $p$, a budget $B$, and a bundle $x^*(p,B,v)$ chosen by the agent given $p$ and $B$: i.e. a bundle $x$ that maximizes $v(x)$ subject to $x\cdot p \leq B$.
\end{definition}

\begin{definition}
An algorithm $A$ $\delta$-learns a class of value functions $\mathcal{V}$ from $m = m(\delta)$ observations if for every distribution $\mathcal{D}$ over examples and for every value function $v \in \mathcal{V}$, given a set of $m$ observations $\{(p_i,B_i,x^*(p_i,B_i,v))\}_{i=1}^m$ where examples $(p_i,B_i)$ are drawn i.i.d. from $\mathcal{D}$, with probability $1-\delta$ it produces a hypothesis $\hat{v}$ such that:
$$\Pr_{(p,B) \sim \mathcal{D}}[v(x^*(p,B,v)) = v(x^*(p,B,\hat{v}))] \geq 1-\delta.$$
We say that $A$ is \emph{efficient} if both its run-time and its sample complexity $m(\delta)$ are bounded by some polynomial $p(n,1/\delta)$. We say that the sample complexity of learning $\mathcal{V}$ is at most $m^* = m^*(\delta)$ if there is some algorithm $A$ which $\delta$-learns $\mathcal{V}$ from $m(\delta) \leq m^*(\delta)$ observations.
\end{definition}

\begin{remark}
Note that a learning algorithm must with high probability (over the choice of observations and coins of the mechanism) produce a value function which \emph{most of the time} (over draws of examples) selects a bundle which is equal to the bundle that the agent would have selected.
\end{remark}

In section \ref{sec:sampling} we relax our definition of learning to allow our learning algorithm to predict bundles which are only \emph{approximately optimal} to the agent, rather than requiring that it select the exactly correct bundle. Note that such approximately optimal bundles might look very different from exactly optimal bundles, and so we will also need to allow our learning algorithms to receive richer feedback from the agent.

\begin{definition}
An algorithm $A$ $(\epsilon, \delta)$-learns a class of value functions $\mathcal{V}$ from $m = m(\delta)$ observations if for every distribution $\mathcal{D}$ over examples and for every value function $v \in \mathcal{V}$, given a set of $m$ observations $\{(p_i,B_i,x^*(p_i,B_i,v))\}_{i=1}^m$ where examples $(p_i,B_i)$ are drawn i.i.d. from $\mathcal{D}$, with probability $1-\delta$ it produces a hypothesis $\hat{v}$ such that:
$$\Pr_{(p,B) \sim \mathcal{D}}[v(x^*(p,B,v)) \geq v(x^*(p,B,\hat{v})) - \epsilon] \geq 1-\delta.$$
For this notion of additive approximation to be meaningful, we will typically normalize the target utility function $v$ to lie in the range $[0,1]$.
\end{definition} 

\section{All Pairs Comparisons Algorithm: Learning Linear Valuation Functions}
\label{section:APCA}
In this section, we present an algorithm that efficiently $\delta$-learns the class of all linear valuation functions given a set of $m = O(\frac{n^2\ln(n^2/\delta)}{\delta})$ observations. In particular, this provides a quadratic upper bound on the optimal sample complexity $m^*(\delta)$ for learning linear valuation functions. We note that a linear $\Omega(m)$ lower bound is immediate in this setting. We start by characterizing the optimal bundle for an agent maximizing a linear utility function, and give intuition for our learning algorithm.

Let $v^*$ and $p$ denote some fixed value and price vectors respectively, and let $B$ denote some fixed budget. We denote the optimal bundle (according to the linear utility function defined by value vector $v^*$, price vector $p$, and budget $B$) by $x^*$. Recall that the value of the optimal bundle is $v^* \cdot x^*$, and its cost $p \cdot x^*$ is at most the budget $B$. Observe that in choosing bundle $x^*$, the agent is solving a divisible knapsack problem, and so the following structural lemma is immediate.

\begin{lemma}
For any pair of goods $i, j \in [n]$ with $x^*_i > x^*_j$, it must be that:
$$\frac{v^*_i}{p_i} \geq \frac{v^*_j}{p_j}$$
\end{lemma}
Equivalently, for any pair of goods with $\frac{v^*_i}{v^*_j} \geq \frac{p_i}{p_j}$, the optimal bundle ``prefers'' good $i$ over good $j$ (It will never buy any of good $j$ until it has exhausted the supply of good $i$). Our algorithm is based on this structural characterization, and operates by maintaining upper and lower bounds on each of the $n^2$ ratios $\frac{v^*_i}{v^*_j}$ for $i\neq j \in [n]$. Based on this transitive relation, we can sort the goods, and find the optimal bundle by buying the goods one by one starting from high priority goods until the budget $B$ is spent completely. In this optimal bundle, we have at most one fractional item.
In our algorithm, we try to learn ratios $\frac{v_i}{v_j}$ accurately for all pair of goods with high probability.

\begin{figure}[h]
\begin{boxedminipage}{\textwidth}
\noindent \textbf{AllPairsLearn}$(\delta)$:
\newline\textbf{Training Phase:}
\begin{enumerate}
\item \textbf{Let} $E$ be a set of $m = O\left(\frac{n^2\ln(n^2/\delta)}{\delta}\right)$ observations  $(p, B, x^*(p, B, v))$.
\item \textbf{Initialize} bounds $(L_{i,j}, U_{i,j})$ for each $i \neq j \in [n]$. Initially $L_{i,j} = 0$ and $U_{i,j} = \infty$ for all $i,j$.
\item \textbf{For} each $(p, B, x^*) \in E$:
    \begin{enumerate}
    \item \textbf{For} each $i \neq j \in [n]$:
        \begin{enumerate}
        \item \textbf{If} $x_i^* > x_j^*$, \textbf{Let} $L_{i,j} = \max(L_{i,j}, \frac{p_i}{p_j})$
        \item \textbf{If} $x_j^* > x_i^*$, \textbf{Let} $U_{i,j} = \min(U_{i,j}, \frac{p_i}{p_j})$
        \end{enumerate}
    \end{enumerate}
\end{enumerate}
\textbf{Classification Phase:}
\begin{enumerate}
\item On a new example $(p, B)$ let $v' \in [0,1]^n$ be any vector such that for all $i\neq j \in [n]$ $\frac{v'_i}{v'_j} \in [L_{i,j}, U_{i,j}]$. Predict bundle $x'(p, B, v')$ that results from maximizing $v'$ with respect to prices $p$ and budget constraint $B$.
\end{enumerate}
\end{boxedminipage}
\caption{The All Pairs Comparison Algorithm for Learning Linear Valuation Functions. It takes as input an accuracy parameter $\delta$.}
\end{figure}


The intuition is that in order to find the optimal bundle $x^*$, we need only know bounds on the ratios of the values of pairs of goods for which unequal quantities are purchased in the optimal bundle. So if we know that $\frac{v_i}{v_j} \geq \frac{p_i}{p_j}$ for any pair of goods with $x^*_i > x^*_j$, we can find the optimal bundle $x^*$. We need not know the values themselves -- it is sufficient to bound these ratios. For example, if the lower bound $L_{i,j}$ is at least $\frac{p_i}{p_j}$, we can infer that good $i$ is preferred to good $j$. If we can infer all these preferences for pairs of goods $(i,j)$ with $x^*_i \neq x^*_j$, we can find the optimal bundle as well. Following we show that with high probability after observing $m=O(n^2\ln(n^2/\delta)/\delta)$ i.i.d. examples we can find the optimal bundle.

\begin{theorem}\label{lemma:allpairsratios}
AllPairsLearn$(\delta)$ efficiently $\delta$-learns the class of linear valuation functions given $m = O\left(\frac{n^2\ln(n^2/\delta)}{\delta}\right)$ observations.
\end{theorem}

\begin{proof}
For each pair of goods $(i,j)$, we define $a_{i,j}$ and $b_{i,j}$ as follows:

\begin{eqnarray*}
a_{i,j} &=& \min \left\{ a | a \leq \frac{v_i}{v_j} ~~~~ \& ~~~~ Pr\left( x^*_i > x^*_j ~~~~ \& ~~~~ \frac{p_i}{p_j} \in [a,\frac{v_i}{v_j}] \right) \leq \frac{\delta}{n^2} \right\} \\
b_{i,j} &=& \max \left\{ b | b \geq \frac{v_i}{v_j} ~~~~ \& ~~~~ Pr\left( x^*_j > x^*_i ~~~~ \& ~~~~ \frac{p_i}{p_j} \in [\frac{v_i}{v_j},b] \right) \leq \frac{\delta}{n^2} \right\} \\
\end{eqnarray*}

where $p$ is the price vector drawn from the distribution $\mathcal{D}$, and $x^*$ is its optimal bundle. Every time an i.i.d. example is drawn, with probability $\delta/n^2$, the lower bound $L_{i,j}$ becomes at least $a_{i,j}$, and the upper bound $U_{i,j}$ becomes at most $b_{i,j}$ for every pair $(i,j)$. For each pair $(i,j)$ after $m$ observations, $L_{i,j}$ is less than $a_{i,j}$ with probability at most $(1-\delta/n^2)^m \leq e^{-\ln(n^2/\delta)} \leq \delta/n^2$. A similar argument holds for $U_{i,j}$.  Using union bound, we can have that with probability $1-\delta$, every $L_{i,j}$ is at least $a_{i,j}$, and every $U_{i,j}$ is at most $b_{i,j}$.

Now when a new example $(p',B',x'(p',B',v))$ arrives ($x'$ is the optimal bundle), the probability that $x'_i \neq x'_j$ and we can not imply which of these two items are preferred over the other one, i.e. $\frac{p_i}{p_j} \in [L_{i,j}, U_{i,j}]$ is at most $2\delta/n^2$, because we know that $[L_{i,j},U_{i,j}] \subseteq [a_{i,j}, b_{i,j}]$. Using union bound, with probability $1-\delta$ we can derive all preference relations for items with unequal fractions in the optimal bundle $x'$. In the other words, with probability $1-\delta$,  we can find the optimal bundle $x'$.

\end{proof}

\section{Learning Linearly Separable Concave Utility Function}
In this section, we modify the algorithm presented in section~\ref{section:APCA} to learn the class of linearly separable concave utility functions. Recall that agents with linearly separable utility functions have a separate function $v_i:[0,1] \rightarrow \bR_+$ for each $1 \leq i \leq n$, and their utility for bundle $x$ is $\sum_{i=1}^n v_i(x_i)$.
We assume that each utility function $v_i$ is a concave function with bounded second derivative. Concavity corresponds to a decreasing marginal utility condition: that buying an additional $\epsilon$ fraction of item $i$ increases agent  utility more when we have less of item $i$: $v_i(a+\epsilon) - v_i(a) \geq v_i(b+\epsilon) - v_i(b)$ for any $a \leq b$. Our bounded second derivative assumption states that the second derivative of each utility function has some supremum strictly less than $\infty$.

We first characterize optimal bundles, and then adapt our learning algorithm for linear valuation functions to apply to the class of linearly separable concave utility functions.

Fix a utility function $v^* = \{v^*_i:[0,1]\rightarrow \bR^+\}$ and a price/budget pair $(p,B)$. The corresponding optimal bundle can be characterized as follows. For any threshold $\tau \geq 0$, define $x^{\tau}_i$ to be $Max\{f | f \in [0,1] \& \frac{v'_i(f)}{p_i} \geq \tau \}$ where $v'_i(f)$ is the first derivative of function $v_i$ at point $f$. We can now define $p^{\tau}$ to be $\sum_{i=1}^n p_ix^{\tau}_i$. We will show that the optimal bundle $x^*$ for $v^*$ in the face of price/budget pair $(p,B)$ is the vector such that $x^*_i = x^{\tau}_i$ for each $1 \leq i \leq n$ where $\tau$ is the maximum value such that this bundle does not exceed the budget constraint.
The following lemma is proved in Appendix~\ref{Appendix:OmittedProofs}.

\begin{lemma}\label{Lemma:OptimalBundleAdditivelySeparable}
The optimal bundle $x^*$ for pair $(p,B)$ is equal to $x^{\tau}$ where $\tau$ is $Max\{\tau | p^{\tau} \leq B \}$.
\end{lemma}

The intuition for our algorithm now follows from the linear utility case.
From each observation consisting of an example and its optimal bundle, we may infer some constraints on the derivatives of utility functions at various points. Just as in the linear utility case, these are the only pieces of information we need to infer the optimal bundle.

\begin{figure}[h]
\begin{boxedminipage}{\textwidth}
\noindent \textbf{LinearSeparableLearn}$(\epsilon,\delta)$:
\newline\textbf{Training Phase:}
\begin{enumerate}
\item \textbf{Let} $E$ be a set of $m = O\left(\frac{(n(k+2))^2\ln((n(k+2))^2/\delta)}{\delta}\right)$ observations  $(p, B, x^*(p, B, v))$.
\item \textbf{Initialize} bounds $(L(i,r,j,s), U(i,r,j,s))$ for each $i \neq j \in [n]$ and $r, s \in [k]$ defined in Definition~\ref{Def:AddSep}. Initially $L(i,r,j,s) = 0$ and $U(i,r,j,s) = \infty$. 
\item \textbf{For} each $(p, B, x^*) \in E$:
    \begin{enumerate}
    \item \textbf{For} each $i \neq j \in [n]$:
        \begin{enumerate}
        \item \textbf{If} $x_i^* > x_j^*$, \textbf{Let} $L(i, \lfloor kx^*_i \rfloor,j,\lceil kx^*_j \rceil) = \max(L(i, \lfloor kx^*_i \rfloor,j,\lceil kx^*_j \rceil), \frac{p_i}{p_j})$
        \item \textbf{If} $x_i^* > x_j^*$, \textbf{Let} $U(i, \lceil kx^*_i \rceil,j,\lfloor kx^*_j \rfloor) = \min(U(i, \lceil kx^*_i \rceil,j,\lfloor kx^*_j \rfloor), \frac{p_i}{p_j})$
    \end{enumerate}
\end{enumerate}
\end{enumerate}
\textbf{Classification Phase:}
\begin{enumerate}
\item  On a new example $(p, B)$ find thresholds $\{l_i\}_{i=1}^n$ such that $\frac{v'_i(l_i/k)}{v'_j((l_j+1)/k)} \geq \frac{p_i}{p_j}$ for each pair $i,j \in [n]$, and $\sum_{i=1}^n \frac{p_iMax\{l_i,0\}}{k} \leq B \leq  \sum_{i=1}^n \frac{p_iMin\{l_i+1,k\} }{k}$. Buy $l_i/k$ fraction of object $i$ for every $i \in [n]$, and spend the remaining budget to buy equal fraction of all objects.
\end{enumerate}
\end{boxedminipage}
\caption{The Learning Algorithm for Linearly Separable Valuation Functions. It takes as input an accuracy parameter $\delta$, and an error parameter $\epsilon$.}\end{figure}

Unlike the linear utility setting, however, it is not possible to maintain bounds on all ratios of derivatives of utility functions at all relevant points,
because there are a continuum of points and the derivatives may take a distinct value at each point.
Instead, we discretize the interval $[0,1]$ with $k+1$ equally distanced points $0, 1/k, 2/k, \cdots, 1$ for some positive integer value of $k$,
and maintain bounds on the ratios of the derivatives at these points.

\begin{definition}\label{Def:AddSep}
We let $k$ to be an integer at least $\big{[} (2Q/\epsilon) \cdot \max_{(p,B) \sim \mathcal{D}, 1 \leq j \leq n} \{\frac{B}{p_j} \} \big{]}$
where $Q$ is an upper bound on $v_i''(x)$ over all $i$ and $x \in [0,1]$, and $\epsilon$ is the error with which we are happy learning to.
We define $V(i,l) = v'_i(l/k)$ for item $i$, $1 \leq i \leq n$ and discretization step $l$, $0 \leq l \leq k$.
For convenience, we define $V(i,k+1) = 0$. For any pairs $1 \leq i, j \leq n$, and $0\ \leq r, s \leq l$, we define $L(i,r,j,s)$ and $U(i,r,j,s)$ to be the lower and upper bounds
on the ratio $\frac{V(i,r)}{V(j,s)}$. The lower and upper bounds are intialized to zero and $\infty$ respectively.
\end{definition}

Analogously to the linear case, our algorithm will maintain upper and lower bounds on the pairwise ratios between each of these these $n(k+2)$ variables.
Since the utilities are concave, we will also maintain the constraint that $V(i,l) \leq V(i,l-1)$ for any $1 \leq i \leq n$ and $1 \leq l \leq k+1$ throughout the course of the algorithm.

In the training phase, the algorithm selects $m=O((n(k+2))^2 log((n(k+2))^2/\delta) /\delta)$ observations. Note the similarity in the number of examples here as compared to the linear case: this is no coincidence. Instead of maintaining bounds on the pairwise ratios of $n$ derivatives we are maintaining bounds on the pairwise ratios between  $n(k+2)$ derivatives.

Consider the inequalities we can infer from each observation $(p, B, x^*)$. By our optimality characterization, we know that for any pair of items $i$ and $j$ with $x^*_i > 0$ and  $x^*_j < 1$, we must have: $\frac{v'_i(x^*_i)}{p_i} \geq  \frac{v'_j(x^*_j)}{p_j}$. We therefore can obtain the following inequality:

$$
\frac{V(i,\lfloor kx^*_i \rfloor)}{p_i} \geq \frac{v'_i(x^*_i)}{p_i} \geq  \frac{v'_j(x^*_j)}{p_j} \geq \frac{V(j,\lceil kx^*_j \rceil )}{p_j}
$$

The above inequality defines the update step that we can impose on the lower bound $L(i,l',j,l'')$ and upper bound $U(i,l',j,l'')$ on the ratios $\frac{V(i,l')}{V(j,l'')}$ where $l'=\lfloor kx^*_i \rfloor$, and $l''=\lceil kx^*_j \rceil$ , analogously to our algorithms update for the linear case. For each example, we update these bounds appropriately.

After the training phase completes, our algorithm uses these bounds to predict a bundle for a new example $(p, B)$. The algorithm attempts to find some threshold $-1 \leq l_i \leq k$ for each item $i$
such that the following two  properties hold.  We define $V(i,-1)= \infty$ for each $1 \leq i \leq n$.
\begin{itemize}
\item
For each pair of items $i\neq j \in [n]$, upper and lower bounds imply that $\frac{V(i,l_i)}{p_i} \geq \frac{V(j,l_j+1)}{p_j}$.
\item
We have that: $\sum_{i=1}^n \frac{p_iMax\{l_i,0\}}{k} \leq B \leq  \sum_{i=1}^n \frac{p_iMin\{l_i+1,k\} }{k}$. In other words, there is enough budget to buy $\max\{l_i,0\}/k$ fraction of object $i$ for all $1 \leq i \leq n$, and the total cost of  buying $\min\{l_i+1,k\}/k$ fraction of each item $i$ is at least $B$.
\end{itemize}

After finding these thresholds $l_1, l_2, \cdots, l_n$, our algorithm selects a bundle that contains $\max\{l_i,0\}/k$ units of item $i$ for each $i$, and then spend the rest of the budget (if there is any remaining) to buy an equal fraction of all objects with $0 \leq l_i < k$, i.e. if $B'$ of the budget remains after the first step, we buy $\frac{B'}{\sum_{1 \leq i \leq n, \ 0 \leq l_i < k} p_i}$ units of each object $i$ with  $0 \leq l_i < k$. Intuitively, the objects with $l_i =0 $, represent very expensive objects (in comparison to their values) which we prefer not to buy at all. On the other hand, we have already exhausted the supply of objects with $l_i=1$.

In the rest of this section, we show in Lemma~\ref{Lemma:FindThresholds} (proved in Appendix~\ref{Appendix:OmittedProofs})
how to find these thresholds (the sequence $l_i$ for $1 \leq i \leq n$) based on the learned upper and lower bounds on ratios if such thresholds exist.
Then, we prove in Lemma\ref{lemma:thresholdswhp} that after training on $m$ examples, with high probability (at least $1-2\delta$), this sequence of thresholds indeed exists. Finally we conclude that our algorithm is an $(\epsilon,\delta)$-learner.

\begin{lemma}\label{Lemma:FindThresholds}
Assuming there exists a sequence of thresholds $\{l_i\}_{i=1}^n$ with the two desired properties in our algorithm, there exists a polynomial time algorithm to find them.
\end{lemma}

We now prove that the required sequence of thresholds $\{l_i\}_{i=1}^n$ exist with high probability. The proof is very similar to Lemma~\ref{lemma:allpairsratios}, and included in Appendix \ref{Appendix:OmittedProofs}.

\begin{lemma}\label{lemma:thresholdswhp}
After updating the algorithm's upper and lower bounds using $m=O((n(k+2))^2 log((n(k+2))^2/\delta) /\delta)$ observations, when considering a new example $(p, B)$, the sequence of thresholds $\{l_i\}_{i=1}^n$ exists with probability at least $1-2\delta$.
\end{lemma}

To conclude, we just need to show that if we find the thresholds with the desired properties, the returned bundle is a good approximation of the optimum bundle. The proof can be found in Appendix~\ref{Appendix:OmittedProofs}.

\begin{theorem}\label{Theorem:AdditivelySeparable}
For any $\epsilon > 0$, we can find some  $k$ (the discretization factor) such that with probability at least $1-2\delta$ over the choice of example $(p, B)$, the bundle $\hat{x} = \hat{x}(p, B)$ returned by our mechanism admits at least one of the following properties:
\begin{enumerate}
\item For each item $1 \leq i \leq n$, we have that $\hat{x}_i \geq x^*_i-\epsilon$,
\item $v^*(\hat{x}) \geq v^*(x^*) - \epsilon$
\end{enumerate}
In other words, have that our mechanism is an efficient $(\epsilon,\delta)$-learning algorithm for the class of linearly separable concave utility functions with bounded range $v:[0,1]^n\rightarrow [0,1]$.
\end{theorem}

\section{A Learning Algorithm based on Sampling from a Convex Polytope}
\label{sec:sampling}
In this section, we present another learning algorithm for $(\epsilon,\delta)$-learning linear cost functions. We introduce a new model, that gets a stronger form of feedback from the agent, and as a result achieve an improved sample complexity bound that requires only $m = \tilde{O}\left(\frac{n\textrm{polylog}(n)}{\delta^3}\right)$ observations.

During the training phase of our algorithm, it will interact with the agent by adding constraints to a linear program and given a new example, propose a candidate bundle to the agent. The agent will either accept the candidate bundle (if it is approximately optimal), or else return to the algorithm a set of linear constraints witnessing the suboptimality of the proposed bundle. The main idea is that for each new example either our algorithm's bundle is almost optimal, or we receive a set of linear constraints to add to our linear program that substantially reduce the volume of the feasible polytope. If the set of constraints are restrictive enough, with high probability, we achieve an approximately optimum bundle on all new examples, and we can end the training phase. Otherwise each new example cuts off some constant fraction of the linear program polytope with high probability. After feeding a polynomial number of examples, and using some arguments to upper bound the volume of the polytope at the beginning and lower bound its volume at the end, we can prove with high probability, the algorithm finds an almost optimal bundle for future examples. First we explain the model, and then we present our algorithm.

\textbf{ Model:} We consider agents with linear utility functions, here bounded so that $v \in [0,1]^n$ . If we have that $v \cdot \hat{x} \geq v \cdot x - \epsilon$, we say bundle $\hat{x}$ is an $\epsilon$-additive approximation to the optimal bundle $x^* = x^*(v^*, p, B)$, and it will be accepted by the agent if it is proposed. If a proposed bundle $\hat{x}$ is not $\epsilon$-approximately optimal, the agent rejects the bundle if proposed, and instead returns a set of inequalities which are witness to the sub-optimality of our solution. The agent returns all valid inequalities of the following form for different pairs of objects $i,j \in [n]$:  $\frac{v_i-\epsilon'}{p_i} > \frac{v_j+\epsilon'}{p_j}$ where $\epsilon' = \epsilon/nM$, and $M$ is the maximum ratio of two different prices in the domain of the price distribution ($\mathcal{D}$).

Intuitively, for these pairs we have that $\frac{v_i}{p_i}$ is greater than $\frac{v_j}{p_j}$ by some non-negligible margin.
In the following, we show that for any suboptimal bundle (not an $\epsilon$-additive approximation) resulted from a value vector $\hat{v}$, there exists at least one of these inequalities  for which we have that $\frac{\hat{v}_i}{p_i} \leq \frac{\hat{v}_j}{p_j}$.
In other words, these set of inequalities that our algorithm returns could be seen as some evidence of suboptimality for any suboptimal bundle for example $(p,B)$.

\begin{lemma}\label{lemma:suboptimalvalues}
For any pair of price vector and budget $(p,B)$, and a suboptimal sampled value vector $\hat{v}$ (that does not generate an $\epsilon$-approximately optimal bundle $\hat{x}$), there exists at least one pair of items $(i,j)$ such that we have $\frac{v_i-\epsilon'}{p_i} > \frac{v_j+\epsilon'}{p_j}$,  and $\frac{\hat{v}_i}{p_i} \leq \frac{\hat{v}_j}{p_j}$.
\end{lemma}

\begin{proof}
Let $x^*$ and $\hat{x}$ be the optimal bundle and the returned bundle based on $\hat{v}$ respectively.
We note that since all objects have non-negative values, we have that $x^* \cdot p = \hat{x} \cdot p = B$ unless the budget $B$ is enough to buy all objects in which case both $x^*$ and $\hat{x}$ are equal to $(1, 1, \cdots, 1)$ which is a contradiction because we assumed $\hat{x}$ is suboptimal.

We can exchange $v/p_i$ units of object $i$ with $v/p_j$ units of item $j$ and vice versa without violating the budget constraint. We show that all the differences in entries of $x^*$ and $\hat{x}$ can be seen as the sum of at most $n$ of these simple exchanges between pairs of objects as follows. We take two entries $i$ and $j$ such that $x^*_i > \hat{x}_i$ and $x^*_j < \hat{x}_j$. We note that as long as two vectors $x^*$ and $\hat{x}$ are not the same, we can find such a pair because we also have that $v \cdot p = \hat{v} \cdot p$.  Without loss of generality, assume that $(x^*_i-\hat{x}_i)p_i \leq (\hat{x}_j-x^*_j)p_j$. Now we buy $x^*_i-\hat{x}_i$ more units of item $i$ in bundle $\hat{x}$ to make the two entries associated with object $i$ in bundles $x^*$ and $\hat{x}$ equal. Instead we buy $(x^*_i-\hat{x}_i)p_i/p_j$ fewer units of object $j$ to obey the budget limit $B$. This way, we decrease the number of different entries in $x^*$ and $\hat{x}$, so after at most $n$ exchanges we make $\hat{x}$ equal to $x^*$. By assumption, $v^*(\hat{x}) \leq v^*(x^*)-\epsilon$. Therefore, in at least one of these exchanges, the value of $\hat{x}$ is increased by more than $\epsilon/n$.

Assume this increase happened in exchange of objects $i$ and $j$. Let $r$ be $(x^*_i-\hat{x}_i)p_i$. We bought $r/p_i$ more units of $i$, and $r/p_j$ fewer units of $j$. The increase in value is $r(v_i/p_i-v_j/p_j) = (x^*_i-\hat{x}_i)(v_i-v_jp_i/p_j) \geq \epsilon/n$. Since $x^*_i-\hat{x}_i$ is at most $1$, we also have that $v_i-v_jp_i/p_j > \epsilon/n$ which can be rewritten as: $v_i- \epsilon/2n > v_jp_i/p_j + \epsilon/2n$. This is equivalent to $\frac{v_i- \epsilon/2n}{p_i} > \frac{v_j + (\epsilon/2n)(p_j/p_i)}{p_j}$. We can conclude that $\frac{v_i- \epsilon/(2nM)}{p_i} > \frac{v_j + (\epsilon/2nM)}{p_j}$ which is by definition of $\epsilon'$: $\frac{v_i- \epsilon'}{p_i} > \frac{v_j + \epsilon'}{p_j}$.

We also note that $\hat{x}_i < 1$ and $\hat{x}_j > 0$, so we can infer that $\frac{\hat{v}_i}{p_i} \leq \frac{\hat{v}_j}{p_j}$. Otherwise one could exchange some fraction of $j$ with some fraction of $i$ and gain more value with respect to value vector $\hat{v}$. This completes the proof of both inequalities claimed in this lemma.
\end{proof}

\textbf{Algorithm:} We maintain a linear program with $n$ variables representing a hypothesis value vector $\hat{v}$. Since $v$ is in $[0,1]^n$, we initially have the constraints: $0 \leq v_i \leq 1$ for any $1 \leq i \leq n$. At any given time, our set of constraints forms a convex body $K$.

Our algorithm loops until we reach a desired property. At each step of the loop we sample $\frac{C \log(n)\log(1/\delta)}{\delta^2}$ examples, and for each of them  we sample uniformly at random a vector $\hat{v}$ from the convex body $K$, and predict the optimal bundle based on this sampled vector. (Note that uniform sampling from a convex body can be done in polynomial time by \cite{FK}). At the end of the loop, we add the linear constraints that we obtained as feedback from the agent to our linear program, and get a more restricted version of $K$ which we call $K'$.

If the volume of $K'$ is greater than $1-\delta$ times the volume of $K$, we stop the learning algorithm, and return $K$ as the candidate convex body. Otherwise, we replace $K$ with the new more constrained body $K'$, and repeat the same loop again. To avoid confusion, we name the final returned convex body $\hat{K}$. After the training phase ends, for future examples, our algorithm samples a value vector $\hat{v}$ uniformly at random from this convex body $\hat{K}$, and predicts the optimal bundle based on $\hat{v}$. We explain what kinds of constraints we add at the end of each loop to find $K'$.

Each iteration of the training phase uses  $\frac{C\log(n)\log(1/\delta)}{\delta^2}$ examples. Recall that for each one, the mechanism proposes a bundle to the agent, who either accepts or rejects it. For each rejected bundle, we are given a set of pairs of objects $(i,j)$ such that $\frac{v_i-\epsilon'}{p_i} > \frac{v_j+\epsilon'}{p_j}$. For each inequality like this, we add the looser constraint $\frac{v_i}{p_i} > \frac{v_j}{p_j}$. At the end, we have a more restricted convex body $K'$ which is formed by adding all of these constraints to $K$.

We must show that after the training phase of the algorithm terminates, we are left with a hypothesis which succeeds at predicting valuable bundles with high probability. We must also also  bound the number of iterations (and therefore the number of examples used by the algorithm) before the training phase of the algorithm terminates. First we bound the total number of iterations of the training phase.

\begin{lemma}
\label{lemma:samplingExamples}
The total number of examples sampled by our algorithm is at most
$$m = O\left(\frac{n\log(n)(\log(n)+\log(M))\log(1/\epsilon)\log(1/\delta)}{\delta^3}\right).$$
\end{lemma}

Finally, we argue that after the learning phase terminates, the algorithm returns a good hypothesis.

\begin{theorem}
The algorithm $(\epsilon, \delta)$-learns from the set of linear utility functions.
\end{theorem}

\begin{proof}
Given a new example $(p, B)$, the algorithm samples a value vector $\hat{v}$ uniformly at random from the convex body $\hat{K}$, and returns an optimal bundle with respect to $\hat{v}$, $p$, and $B$.

Consider a price vector $p$ and budget $B$.  For some value vectors in $\hat{K}$, the returned bundle is suboptimal (not an $\epsilon$-additive approximation). We call this subset the set of suboptimal value vectors with respect to $(p,B)$, and the fraction of suboptimal value vectors in $\hat{K}$ is the probability that our algorithm does not return a good bundle, i.e. the error probability of our algorithm.
We say a pair $(p,B)$ is unlucky if for more than $\delta$ fraction of value vectors in $\hat{K}$, the returned bundle is suboptimal.
We prove that with probability at least $1-\delta/2$, the convex body $\hat{K}$ we return, has this property that with at most probability $\delta/2$, the pair $(p,B)$ drawn from $\mathcal{D}$ is unlucky. This way with probability at most $\delta/2+\delta/2=\delta$, the pair $(p,B)$ is unlucky which proves that our algorithm is  $(\epsilon, \delta)$-learner.

We prove the claim by contradiction.
Define $A$ to be the event that "with probability more than $\delta/2$, the pair $(p,B) \sim \mathcal{D}$ is unlucky".
 We prove that the probability of event $A$ is at most $\delta/2$.
Let $K_i$ be the convex body at the beginning of iteration $i$, and $K'_i$ be the more restricted version of $K_i$ that we compute at the end of iteration $i$. Event $A$ holds if for some $i$ we have these two properties: a) the probability that a pair $(p,B)$ drawn i.i.d. from $\mathcal{D}$ is unlucky with respect to $K_i$ is more than $\delta/2$, i.e. if we sample the value vector from $K_i$, the returned bundle for $(p,B)$ is suboptimal with probability more than $\delta$. b) the volume of $K'_i$ is not less than $1-\delta$ times volume of $K_i$.

We bound the probability of having both of these properties at iteration $i$. In this iteration, for every example we take, with probability more than $\delta/2$, the pair $(p,B)$ is unlucky. For an unlucky pair $(p,B)$,  with probability more than $\delta$, we return a suboptimal example, and then we get feedback from the agent. Using lemma~\ref{lemma:suboptimalvalues}, and the feedback we get from the agent, all of the suboptimal value vectors for pair $(p,B)$ will be removed from $K_i$ and will not exist in $K'_i$ (by the new constraints we add in this loop). Since $(p,B)$ is unlucky, more than  $\delta$ fraction of the $K_i$ will be deleted in this case. In other words, for each example in loop $i$ with probability at least $\delta^2/2$, more than $\delta$ fraction of $K_i$ will be removed. Clearly, since $K'_i$ has volume at least $1-\delta$ fraction of $K_i$, this has not happened for any of the examples of loop $i$. Since we have $\frac{C\log(n)\log(1/\delta)}{\delta^2}$ examples in each loop, the probability of holding both these properties at loop $i$ is at most $(1-\delta^2)^{\frac{C\log(n)\log(1/\delta)}{\delta^2}} < \delta /(2n^C)$ for $\delta \leq 1/2$. Since there are less than $n^C$ number of loops for some large enough constant $C$, the probability of event $A$ (which might happen in any of the loops) is less than $\delta/2$.
\end{proof}

\section{Discussion}
In this paper we have considered the problem of efficiently learning predictive classifiers from revealed preferences. We feel that the revealed preferences problem is much more meaningful when the observed data must be rationalized with a \emph{predictive} hypothesis, and of course much remains to be done in this study. Our work leaves many open questions:
\begin{enumerate}
\item What are tight bounds on the sample complexity for $\delta$-learning linear valuation functions? There is a simple $\Omega(n)$ lower bound, and here we give an algorithm with sample complexity $\tilde{O}(n^2/\delta)$, but where does the truth lie?
\item Is there a general measure of sample complexity, akin to VC-dimension in the classical learning setting, that can be fruitfully applied to the revealed preferences problem? Beigman and Vohra \cite{BV06} adapt the notion of fat-shattering dimension to this setting, but applied to the revealed preferences problem, fat shattering dimension is cumbersome and seems ill-suited to proving tight polynomial bounds.
\end{enumerate}

\bibliographystyle{alpha}
\bibliography{revealedprefs}

\begin{thebibliography}{BCIW12}

\bibitem[Afr65]{Afr65}
S.N. Afriat.
\newblock The equivalence in two dimensions of the strong and weak axioms of
  reveaded preference.
\newblock {\em Metroeconomica}, 17(1-2):24--28, 1965.

\bibitem[Afr67]{Afr67}
S.N. Afriat.
\newblock The construction of utility functions from expenditure data.
\newblock {\em International Economic Review}, 8(1):67--77, 1967.

\bibitem[BCIW12]{BCIW11}
M.F. Balcan, F.~Constantin, S.~Iwata, and L.~Wang.
\newblock Learning valuation functions.
\newblock In {\em COLT}, 2012.

\bibitem[BH11]{BH11}
M.F. Balcan and N.~Harvey.
\newblock Learning submodular functions.
\newblock In {\em STOC 2011}, pages 793--802, 2011.

\bibitem[BV06]{BV06}
E.~Beigman and R.~Vohra.
\newblock Learning from revealed preference.
\newblock In {\em Proceedings of the 7th ACM Conference on Electronic
  Commerce}, pages 36--42. ACM, 2006.

\bibitem[DFK91]{FK}
M.~Dyer, A.~Frieze, and R.~Kannan.
\newblock A random polynomial-time algorithm for approximating the volume of
  convex bodies.
\newblock {\em Journal of the ACM (JACM)}, 38(1):1--17, 1991.

\bibitem[EGW11]{EGW11}
F.~Echenique, D.~Golovin, and A.~Wierman.
\newblock A revealed preference approach to computational complexity in
  economics.
\newblock In {\em ACM Conference on Electronic Commerce}, pages 101--110, 2011.

\bibitem[Sam38]{Sam38}
P.A. Samuelson.
\newblock A note on the pure theory of consumer's behaviour.
\newblock {\em Economica}, 5(17):61--71, 1938.

\bibitem[Var06]{Var05}
H.R. Varian.
\newblock Revealed preference.
\newblock {\em Samuelsonian economics and the twenty-first century}, pages
  99--115, 2006.

\end{thebibliography}

\appendix
\section{Omitted Proofs}\label{Appendix:OmittedProofs}

\begin{proofof}{Lemma \ref{Lemma:OptimalBundleAdditivelySeparable}}
For each possible pair of items $i\neq j \in [n]$, we consider three cases:
\begin{itemize}
\item $0 < x^*_i , x^*_j < 1$: In this case, $\frac{v'_i(x^*_i)}{p_i} = \frac{v'_j(x^*_j)}{p_j}$. Otherwise (if for example the expression corresponding to item $i$ is greater), we could buy $\epsilon'/p_i$ additional units of $i$, and buy $\epsilon'/p_j$ fewer units of $j$ without violating the budget constraint. When $\epsilon' \to 0$, this exchange will be beneficial for the agent which would contradict optimality.
\item $x^*_i = 1$ and $x^*_j<1$: In this case, $\frac{v'_i(x^*_i)}{p_i} \geq \frac{v'_j(x^*_j)}{p_j}$ otherwise the agent could buy fewer units of $i$ and additional units of $j$ and thereby increase the value of the bundle, contradicting optimality.
\item $x^*_i >0$ and $x^*_j =0$: Identically to above: $\frac{v'_i(x^*_i)}{p_i} \geq \frac{v'_j(x^*_j)}{p_j}$.
\end{itemize}

To complete the proof we need now only to select $\tau$. If there exists some $1 \leq  i \leq n$ such that $0 < x^*_i < 1$, setting $\tau = \frac{v'_i(x^*_i)}{p_i}$ proves the claim. Otherwise, any value $ \tau \in [\max_{i|x^*_i=0} \frac{v'_i(x^*_i)}{p_i}, \min_{i|x^*_i=1} \frac{v'_i(x^*_i)}{p_i}]$ completes the proof.
\end{proofof}

\begin{proofof}{Lemma \ref{Lemma:FindThresholds}}
First let us assume that there is a sequence of thresholds with the desired properties. In this case, we may find it as follows. Suppose item $i$ has the maximum value of $\frac{V(i,l_i+1)}{p_i}$ among all items: i.e. $\frac{V(i,l_i+1)}{p_i} \geq \frac{V(j,l_j+1)}{p_j}$ for any $j \neq i$.  We assume that this item $i$ and threshold $l_i$ are given, because we can guess their values as there are $n(k+2)$ possible choices for them. For any item   $j \neq i$, we select some $l_j$ such that it can be inferred from our upper and lower bounds that $\frac{V(j,l_j)}{p_j} \geq \frac{V(i,l_i+1)}{p_i}$, but it can not be inferred that $\frac{V(j,l_j+1)}{p_j} \geq \frac{V(i,l_i+1)}{p_i}$.

Since we have $V(j,-1)=\infty$ and $V(j,k+1)=0$, we can always find some value for $l_j$. In fact for each item $j$, we can find two thresholds $0 \leq t_1(j) \leq t_2(j) \leq k$ such that a) we can infer that $\frac{V(j,t_1(j))}{p_j} \geq \frac{V(i,l_i+1)}{p_i}$, and b) we can also infer that $\frac{V(j,t')}{p_j} = \frac{V(i,l_i+1)}{p_i}$ for any $t_1(j) < t' \leq t_2(j)$, and finally, c)  we can not infer that $\frac{V(j,t_2(j)+1)}{p_j} \geq \frac{V(i,l_i+1)}{p_i}$. Variable $l_j$ could be any integer in range $[t_1,t_2]$. We might sometimes have that $t_1(j)=t_2(j)$ which means that $l_j$ is uniquely defined.


Assuming object $i$ has the maximum value of $\frac{V(i,l_i+1)}{p_i}$, we know that any solution $l_j \in [t_1(j),t_2(j)]$ (for all $j \neq i$)
satisfies the first property we are looking for. The second property is a budget constraint: we should be able to buy $\max\{l_j,0\}/k$ units of each item $j$, and the total cost of buying $\min\{l_j+1,k\}/k$ of each item  $j$ should be at least $B$.

We start with thresholds $l_j=t_1(j)$. If these are not feasible (i.e. if the resulting bundle costs more than $B$), there does not exist such a sequence of thresholds with object $i$ as the object with maximum $\frac{V(i,l_i+1)}{p_i}$ and $l_i$ as the threshold of object $i$. Alternately, if these thresholds are feasible,  we increase the thresholds one at a time while the cost of the resulting optimal bundle remains below $B$. We have the freedom to increase threshold $l_j$ in the range $[t_1(j),t_2(j)]$, and we can increase it one unit at a time to a maximum of $t_2(j)$. We stop when it is not possible to increase any of the thresholds any more. This process results in a set of thresholds $l_j \in [t_1(j),t_2(j)]$, and it is not possible to increase any of them.

 If for some $j$, $l_j$ is strictly less than $t_2(j)$, we can infer that budget $B$ is not enough to buy $\max\{l_{j'},0\}/k$ units of each object $j' \neq j$, and $(l_{j}+1)/k=Min\{l_j+1,k\}/k$ units of object $j$  (note that $l_j+1 \leq t_2(j) \leq k$) ( Otherwise we could have increased the threshold $l_j$ by at least one). Consequently for this sequence of thresholds, there is not enough budget to buy $\min\{l_{j''}+1,0\}/k$ units of item $j''$ for all $1 \leq j'' \leq n$. Therefore, this sequence satisfies both properties we wanted.

In the remaining case, we stop at $l_{j} = t_2(j)$ for all $j \neq i$. In this case, if cost of buying $\min\{l_{j'}+1,k\}/k$ units of all items $1 \leq j' \leq n$ is at least $B$, this sequence of thresholds $l_1, l_2, \cdots, l_n$ satisfies both properties that we want. Otherwise, we must try another guess for object $i$ and threshold $l_i$ to start again. We try all $n(k+2)$ possible guesses exhaustively for pair $(i,l_i)$, and if in one of them we succeed to find a sequence of valid thresholds, we are done, otherwise there does not exist such sequence, and our algorithm simply returns a random bundle. (The probability that this occurs will be folded into the error probability of our algorithm).
\end{proofof}

\begin{proofof}{Lemma \ref{lemma:thresholdswhp}}
Similar to Lemma~\ref{lemma:allpairsratios},  we define $a_{i,r,j,s}$ and $b_{i,r,j,s}$ where $i$ and $j$ are two objects, and  $0 \leq r,s \leq k$ as follows:

\begin{eqnarray*}
a_{i,r,j,s} &=& \min \left\{ a | Pr\left( r/k \leq x^*_i < (r+1)/k ~~ \& ~~ s/k \leq x^*_j < (s+1)/k ~~ \& ~~ a \leq \frac{p_i}{p_j} \leq \frac{v'_i(x^*_i)}{v'_j(x^*_j)}] \right) \leq \frac{\delta}{(n(k+2))^2} \right\} \\
b_{i,r,j,s} &=& \max \left\{ b | Pr\left( r/k \leq x^*_i < (r+1)/k ~~ \& ~~ s/k \leq x^*_j < (s+1)/k ~~ \& ~~ \frac{v'_i(x^*_i)}{v'_j(x^*_j)} \leq \frac{p_i}{p_j} \leq b] \right) \leq \frac{\delta}{(n(k+2))^2} \right\} \\
\end{eqnarray*}

Every time we see an example, with probability $\frac{\delta}{(n(k+2))^2}$, we update the lower bound $L(i,r,j,s+1)$ to some thing equal to or greater than $a_{i,r,j,s}$. A similar claim holds for the upper bounds and values of $b_{j,s,i,r}$. Similar to the proof of Lemma~\ref{lemma:allpairsratios}, one can show that with probability at most $(1-\delta/(n(k+2))^2)^m \leq \delta/(n(k+2))^2$ the lower bound $L(i,r,j,s+1)$ is less than $a_{i,r,j,s}$ after observing all $m$ examples. Using the union bound this does not happen for any pairs of $(i,r)$ and $(j,s)$ with probability at least $1-\delta$. So with high probability, we have very accurate bounds on the ratios of different first derivatives at points $0, 1/k, 2/k, \cdots, 1$.

In the classification phase of the algorithm, consider a new example is drawn from $\mathcal{D}$. We prove that with probability $1-\delta$, we can find the thresholds $\{l_i\}_{i=1}^n$, if the lower bound $L_{i,r,j,s+1}$ is  at least $a_{i,r,j,s}$ for different choices of $i,j,r$, and $s$. So we assume these inequalities hold.  Define $l_i$ to  be $\lfloor kx^*_i \rfloor$ for $1 \leq i \leq n$ with $x^*_i > 0$, if $x^*_i$ is zero, we define $l_i$ to be $-1$. We know that  $L_{i,l_i,j,l_j+1}$ is at least $a_{i,l_i,j,l_j}$, so we can infer that $\frac{V(i,l_i)}{p_i} \geq \frac{V(j,l_j+1)}{p_j}$ unless $\frac{p_i}{p_j} \in [a_{i,r,j,s}, \frac{v'_i(x^*_i)}{v'_j(x^*_j)}]$ which occurs with probability at most $\delta/(n(k+2))^2$. Taking a union bound again, and considering all choices of $i,j,r,s$, we find that this does not happen for any $4$-tuple $(i,j,r,s)$ except with probability at most $\delta$.

Therefore we have shown that the thresholds
$\{l_i\}_{i=1}^n$ exist and can be inferred based on our upper and lower bounds with probability $1-2\delta$.  Thus our lower bounds are enough to imply the the first property of the thresholds $\{l_i\}_{i=1}^n$. We also note that the second property is satisfied because of the choices of $l_i$. Clearly $\max\{l_i,0\}/k$ is at most $x^*_i$, and therefore the total cost of buying  $\max\{l_i,0\}/k$ of each object $i$ is at most the cost of optimum bundle,  which is $B$. We also know that $\min\{l_i+1,k\}/k$ is at least $x^*_i$ which gives us the remaining inequality needed for the second property of the thresholds.
\end{proofof}

\begin{proofof}{Theorem \ref{Theorem:AdditivelySeparable}}
Using Lemma~\ref{lemma:thresholdswhp},  we know that for any example $(p,B)$, our algorithm finds thresholds $\{l_i\}_{i=1}^n$ with probability $1-2\delta$. We  prove that our bundle has one of the two properties in the statement of this theorem in these cases (when our algorithm finds appropriate thresholds). There are two cases:
\begin{itemize}
\item
There exists some item $i$ such that $x^*_i \geq (l_i+1)/k$. We have that $\frac{V(j,l_j)}{p_j} \geq \frac{V(i,l_i+1)}{p_i} \geq \frac{v'_i(x^*_i)}{p_i}$ for any $j \neq i$. Therefore $x^*_j \geq l_j/k$ for each $1 \leq j \leq n$.

Therefore the proposed bundle is completely consistent with the optimum solution up to $l_{i'}/k$ fraction for each item $1 \leq i' \leq n$. We spend the rest of our budget to buy equal fractions of objects with $0 \leq l_{i''} < k$, but the optimum algorithm might do something else. Based on the second property of thresholds, the remaining budget is not enough to buy more than $1/k$ fraction of these objects, so in the second step we buy some fraction $\rho \leq 1/k$ of all objects with $0 \leq l_{i''} < k$ to spend our budget completely.

To compare the performance of the optimum algorithm and our algorithm on the remaining budget, we will focus on some small part of the remaining budget. With very small budget $\epsilon'>0$, we might buy some fraction of object $j$ (with $0 \leq l_{j} < k$) to increase its quantity by $\epsilon'/p_j$. We know that since the value function for object $j$, $v_j$, is concave with bounded second derivative, our increase in value is at least $(v'_j(l_j/k)-Q/k)\epsilon'/p_j$ where $Q$ is an upper bound on all values of second derivatives of value functions. Because the fraction of object $j$, when we are increasing it, lies in the range $[l_j/k,(l_j+1)/k]$, and clearly has difference at most $1/k$ from fraction $l_j/k$.
Therefore the first derivative of the value function of object $j$ can not be less than $v'_j(l_j/k)-Q/k$, when we are increasing it.

 On the other hand, the optimum solution might use this $\epsilon'$ budget to buy $\epsilon'/p_{j'}$ fraction of object $j'$. Since the fraction object $j'$ is in range $[l_{j'},1]$ when the optimum is buying from $j'$, the increase in the value can not be more than $\big{(}v'_{j'}((l_{j'}+1)/k)+Q/k\big{)} \epsilon'/p_{j'}$. This holds because the fraction of object $j'$   is at most $1/k$ less than $(l_{j'}+1)/k$. This means that the optimum solution is gaining at most $\big{(}v'_{j'}((l_{j'}+1)/k)+Q/k\big{)}/p_{j'} - (v'_j(l_j/k)-Q/k)/p_j$ more value per each unit of budget in comparison to our algorithm. Since we have that $\frac{V(j,l_j)}{p_j} \geq \frac{V(j',l_{j'}+1)}{p_{j'}}$, this term (the difference in values per unit of budget) is at most $Q(1/p_j+1/p_{j'})/k$. In order to make the total difference in the values of the two bundles at most $\epsilon$, one just needs to set
$k \geq \big{[} (2Q/\epsilon) \cdot \max_{p,B \sim \mathcal{D}, 1 \leq j \leq n} \frac{B}{p_j} \big{]}$.

\item In the second case, for all $1 \leq i \leq n$, $x^*_i < (l_{i}+1)/k$, and our fraction for this object is at least $l_i/k$, for $k > 1/\epsilon$, the first property in the statement of this theorem holds.
\end{itemize}
\end{proofof}

\begin{proofof}{Lemma \ref{lemma:samplingExamples}}
By construction, each time we update the convex body $K$, we reduce its volume by a factor of $1-\delta$ using $\frac{C\log(n)\log(1/\delta)}{\delta^2}$ examples. So it suffices to show that we will not do these updates more than $O([n(\log(n)+\log(M))\log(1/\epsilon)]/\delta)$ times. We note that the volume of $K$ is $1$ at the beginning. We prove a lower bound on the volume of the final convex body $\hat{K}$ by showing that some points will not be deleted in any of the iterations. We say a value vector $v'$ is close to the actual value vector $v$, if for any $1 \leq i \leq n$, we have that  $|v_i-v'_i| \leq \epsilon'$. We claim that a vector $v' \in [0,1]^n$ which is close to $v$ will not be removed in any of the loops. We prove by contradiction.

Suppose $v'$ has been removed by adding some constraint on pair of objects $(i,j)$ with price vector $p$. We should have that $\frac{v'_i}{p_i} \leq \frac{v'_j}{p_j}$. Since we added this constraint, we also should have that $\frac{v_i- \epsilon'}{p_i} > \frac{v_j + \epsilon'}{p_j}$. But this is a contradiction since $v'_i \geq v_i- \epsilon'$ and $v'_j \leq v_j + \epsilon'$. So we never remove points from the set $(v+[-\epsilon',\epsilon]^n) \cap [0,1]^n$.

We note that for each $1 \leq i \leq n$, the length of interval $[v_i-\epsilon',v_i+\epsilon'] \cap [0,1]$ is at least $\epsilon'$, so the volume of the set of points $(v+[-\epsilon',\epsilon]^n) \cap [0,1]^n$ is at least $(\epsilon')^n$ which is a lower bound on the volume $\hat{K}$. Therefore the number of iterations can not be more than $\log_{1-\delta} \big{(} \frac{(\epsilon')^n}{1} \big{)} = \frac{\ln((\epsilon')^n)}{\ln(1-\delta)} \leq \frac{n\ln(1/\epsilon')}{\delta}$. By definition of $\epsilon'$, the number of loops is $O(\frac{n(\log(n)+\log(M))log(1/\epsilon)}{\delta})$. So the total number of examples we use to learn $\hat{K}$  is $O(\frac{Cn\log(n)(\log(n)+\log(M))log(1/\epsilon)\log(1/\delta)}{\delta^3})$.
\end{proofof}
\end{document}